\documentclass[a4paper,UKenglish]{lipics-v2016}

\usepackage{framed}
\usepackage{float}
\usepackage{tikz}
\usepackage{adjustbox}

\usepackage{microtype}
\usepackage[linesnumbered,ruled,vlined]{algorithm2e}

\usetikzlibrary{decorations.pathmorphing}
\usetikzlibrary{shapes.geometric}
\usetikzlibrary{shapes.misc}
\usetikzlibrary{patterns}

\pgfdeclarepatternformonly{del lines}{%
    \pgfqpoint{-1pt}{-1pt}}{\pgfqpoint{21pt}{21pt}}{\pgfqpoint{16pt}{16pt}}%
    {
        \pgfsetlinewidth{0.4pt}
        \pgfpathmoveto{\pgfqpoint{0pt}{0pt}}
        \pgfpathlineto{\pgfqpoint{16pt}{16pt}}
        \pgfusepath{stroke}
    }
\pgfdeclarepatternformonly{added lines}{\pgfpointorigin}{\pgfqpoint{50pt}{10pt}}{\pgfqpoint{50pt}{10pt}}%
    {
        \pgfsetlinewidth{0.2pt}
        \pgfpathmoveto{\pgfqpoint{0pt}{0.5pt}}
        \pgfpathlineto{\pgfqpoint{100pt}{0.5pt}}
        \pgfusepath{stroke}
    }

\title{Faster Algorithms for Mean-Payoff Parity Games}

\author[1]{Krishnendu Chatterjee}
\author[2]{Monika Henzinger}
\author[3]{Alexander Svozil}
\affil[1]{IST Austria\\
  \texttt{krish.chat@ist.ac.at}}
\affil[2]{University of Vienna, Faculty of Computer Science, Vienna, Austria\\
  \texttt{monika.henzinger@univie.ac.at}}
\affil[3]{University of Vienna, Faculty of Computer Science, Vienna, Austria\\
  \texttt{alexander.svozil@univie.ac.at}}
\authorrunning{K. Chatterjee, M. Henzinger and A. Svozil} 

\Copyright{Krishnendu Chatterjee, Monika Henzinger and Alexander Svozil}

\EventEditors{Kim G. Larsen, Hans L. Bodlaender, and Jean-Francois Raskin}
\EventNoEds{3}
\EventLongTitle{42nd International Symposium on Mathematical Foundations of
Computer Science (MFCS 2017)}
\EventShortTitle{MFCS 2017}
\EventAcronym{MFCS}
\EventYear{2017}
\EventDate{August 21--25, 2017}
\EventLocation{Aalborg, Denmark}
\EventLogo{}
\SeriesVolume{83}
\ArticleNo{39} 

\renewcommand{\O}{\mathcal{O}}

\newcommand{\In}{\mathit{In}}
\newcommand{\Out}{\mathit{Out}}

\newcommand{\val}{\mathit{val}}

\newcommand{\V}{\mathcal{\nu}}
\newcommand{\N}{\mathbb{N}}
\newcommand{\Z}{\mathbb{Z}}
\newcommand{\Q}{\mathbb{Q}}

\newcommand{\Parity}{\mathit{Parity}}

\newcommand{\MeanPayoff}{\mathit{MeanPayoff}}
\newcommand{\MP}{\mathit{MP}}
\newcommand{\MPP}{\mathit{MPP}}
\newcommand{\Inf}{\mathit{Inf}}
\newcommand{\Plays}{\mathit{Plays}}
\newcommand{\restr}{\upharpoonright}
\newcommand{\outc}{\mathit{outcome}}
\newcommand{\attr}{\mathit{Attr}}
\newcommand{\lift}{\mathit{lift}}
\newcommand{\solvempp}{\mathit{SolveMeanPayoffParity}}

\begin{document}

\maketitle

\begin{abstract}
Graph games provide the foundation for modeling and synthesis of reactive processes. 
Such games are played over graphs where the vertices are controlled by two adversarial players.
We consider graph games where the objective of the first player is the 
conjunction of a qualitative objective (specified as a parity condition) 
and a quantitative objective (specified as a mean-payoff condition).
There are two variants of the problem, namely, the {\em threshold} problem where the 
quantitative goal is to ensure that the mean-payoff value is above a threshold, 
and the {\em value} problem where the quantitative goal is to ensure the optimal
mean-payoff value; in both cases ensuring the qualitative parity objective. 
The previous best-known algorithms for game graphs with $n$ vertices, $m$ edges,
parity objectives with $d$ priorities, and maximal absolute reward value $W$ 
for mean-payoff objectives, are as follows:
$\O(n^{d+1} \cdot m \cdot W)$ for the threshold problem, and 
$\O(n^{d+2} \cdot m \cdot W)$ for the value problem. 
Our main contributions are faster algorithms, and the running times of our algorithms 
are as follows: $\O(n^{d-1} \cdot m \cdot W)$ for the threshold problem, 
and $\O(n^{d} \cdot m \cdot W \cdot \log (n\cdot W))$ for the value problem. 
For mean-payoff parity objectives with two priorities, our algorithms match the best-known 
bounds of the algorithms for mean-payoff games (without conjunction with parity objectives).
Our results are relevant in synthesis of reactive systems with both functional 
requirement (given as a qualitative objective) and performance requirement 
(given as a quantitative objective).
\end{abstract}

\section{Introduction}
\noindent{\em Graph games.}
A graph game is played on a finite directed graph with two players, namely, 
player~1 and player~2 (the adversary of player~1).
The vertex set is partitioned into player-1 and player-2 vertices.
At player-1 vertices, player~1 chooses a successor vertex;
and at player-2 vertices, player~2 
does likewise.
The result of playing the game forever is an infinite path through the graph.
There has been a long history of using graph games for modeling and 
synthesizing reactive processes~\cite{BuchiLandweber,PnueliRosner,RamadgeWonham}:
a reactive system and its environment represent the two players, whose states 
and transitions are specified by the vertices and edges of a game graph.
Consequently, graph games provide the theoretical foundation
for modeling and synthesizing reactive processes.

\smallskip\noindent{\em Qualitative and quantitative objectives.}
For reactive systems, the objective is given as a set of desired paths 
(such as $\omega$-regular specifications), or 
as a quantitative optimization objective with a payoff function on the paths.
The class of $\omega$-regular specifications provide a robust framework
to express all commonly used specifications for reactive systems in 
verification and synthesis. 
Parity objectives are a canonical way to express $\omega$-regular objectives~\cite{Thomas97},
where an integer priority is assigned to every vertex, and a path satisfies the 
parity objective for player~1 if the minimum priority visited infinitely often is even.
One of the classical and most well-studied quantitative objectives is the mean-payoff 
objective, where a reward is associated with every edge, and the payoff of a path is 
the long-run average of the rewards of the path.

\smallskip\noindent{\em Mean-payoff parity objectives.}
Traditionally the verification and the synthesis problems were considered
with qualitative objectives.
However, recently combinations of qualitative and quantitative objectives have 
received a lot of attention.
Qualitative objectives such as $\omega$-regular objectives specify
the functional requirements of reactive systems, whereas the quantitative 
objectives specify resource consumption requirements (such as for embedded 
systems or power-limited systems).
Combining quantitative and qualitative objectives is crucial in the 
design of reactive systems with both resource constraints and functional 
requirements~\cite{CAHS03,CJH05,BFLMS08,BCHJ09}.
For example, mean-payoff parity objectives are relevant in synthesis of 
optimal performance lock-synchronization for programs~\cite{CCHRS11}, 
where one player is the synchronizer, the opponent is 
the environment; the performance criteria is specified as mean-payoff objective; 
and the functional requirement (e.g., data-race freedom or liveness) 
as an $\omega$-regular objective.
Mean-payoff parity objectives have been used in several other applications,
e.g., define permissivity for parity games~\cite{BMOU11} and 
robustness in synthesis~\cite{BCGHHJKK14}.

\smallskip\noindent{\em Threshold and value problems.}
For graph games with mean-payoff and parity objectives there are two variants 
of the problem.
First, the {\em threshold} problem, where a threshold $\V$ is given for 
the mean-payoff objective, and player~1 must ensure the parity objective and 
that the mean-payoff is at least $\V$.
Second, the {\em value} problem, where player~1 maximizes the mean-payoff 
value while ensuring the parity objective.
In the sequel of this section, we will refer to graph games with mean-payoff and
parity objectives as mean-payoff parity games.

\smallskip\noindent{\em Previous results.}
Mean-payoff parity games were first studied in~\cite{CJH05},
and algorithms for the value problem were presented.
It was shown in~\cite{CD10a} that the decision problem for mean-payoff parity games
lies in NP $\cap$ coNP (similar to the status of mean-payoff games and parity games).
For game graphs with $n$ vertices, $m$ edges, parity objectives with $d$ priorities, 
and maximal absolute reward value $W$ for the mean-payoff objective, the previous 
known algorithmic bounds for mean-payoff parity games are as follows:
For the threshold problem the results of~\cite{CD10a} give an 
$\O(n^{d+4}\cdot m\cdot d\cdot W)$-time algorithm.
This algorithmic bound was improved in~\cite{BMOU11} where an $\O(n^{d+2}\cdot m\cdot W)$-time 
algorithm was presented for the value problem. 
The result of~\cite{BMOU11} does not explicitly present any other better bound for the 
threshold problem. However, the recursive algorithm of~\cite{BMOU11} uses value mean-payoff games
as a sub-routine, and replacing value mean-payoff games with threshold mean-payoff games gives 
an $\O(n)$-factor saving, and yields an $\O(n^{d+1}\cdot m\cdot W)$-time algorithm for the 
threshold problem for mean-payoff parity games.

\smallskip\noindent{\em Contributions.} 
In this work our main contributions are faster algorithms to solve mean-payoff parity games.
Previous and our results are summarized in Table~\ref{tab:complexity}.

\begin{enumerate}
\item {\em Threshold problem.} We present an $\O(n^{d-1}\cdot m \cdot W)$-time algorithm
for the threshold problem for mean-payoff parity games, improving the previous 
$\O(n^{d+1}\cdot m \cdot W)$ bound.
The important special case of parity objectives with two priorities correspond to B\"uchi 
and coB\"uchi objectives. 
Our bound for mean-payoff B\"uchi games and mean-payoff coB\"uchi games is 
$\O(n \cdot m \cdot W)$, which matches the best-known bound to solve the threshold
problem for mean-payoff objectives~\cite{brim2011}, and improves the previous known $\O(n^3 \cdot m \cdot W)$
bound~\cite{BMOU11}.

\item {\em Value problem.} We present an $\O(n^{d}\cdot m\cdot W \cdot \log (n\cdot W))$-time 
algorithm for the value problem for mean-payoff parity games, improving the previous 
$\O(n^{d+2}\cdot m \cdot W)$ bound.
Our bound for mean-payoff B\"uchi games and mean-payoff coB\"uchi games is 
$\O(n^2 \cdot m \cdot W \cdot \log(n\cdot W))$, which matches the bound of~\cite{brim2011} 
to solve the value problem for mean-payoff objectives, and improves the previous 
known $\O(n^4 \cdot m \cdot W)$ bound.
\end{enumerate}

\smallskip\noindent{\em Technical contributions.}
Our main technical contributions are as follows:
\begin{enumerate}
\item First, for the threshold problem, we present a decremental algorithm for 
mean-payoff games that supports a sequence of vertex-set deletions along with 
their player-2 reachability set. 
We show that the total running time is $\O(n \cdot m \cdot W)$, which matches the 
best-known bound for the static algorithm to solve mean-payoff games.
We show that using our decremental algorithm we can solve the threshold problem for 
mean-payoff B\"uchi games in time $\O(n \cdot m \cdot W)$.

\item Second, for mean-payoff coB\"uchi games, the decremental approach does not work.
We present a new static algorithm for threshold mean-payoff games that identifies subsets $X$ 
of the winning set for player~1, where the time complexity is $\O(|X| \cdot m \cdot W)$, i.e., 
it replaces $n$ with the size of the set identified. We show that with our new static algorithm 
we can solve the threshold problem for mean-payoff coB\"uchi games in time $\O(n \cdot m \cdot W)$. 

\item Finally, we show for all mean-payoff parity objectives, given an algorithm for the 
threshold problem, the value problem can be solved in time $n \cdot \log (n\cdot W)$
times the complexity of the threshold problem. 
\end{enumerate}

\smallskip\noindent{\em Related works.}
The problem of graph games with mean-payoff parity objectives was first 
studied in~\cite{CJH05}. The NP $\cap$ coNP complexity bound was established 
in~\cite{CD10a}, and an improved algorithm for the problem was given 
in~\cite{BMOU11}.
The mean-payoff parity objectives has also been considered in other
stochastic setting such as Markov decision processes~\cite{CD11,CD11b} 
and stochastic games~\cite{CDGO14}.
The algorithmic approaches for stochastic games build on the results for
non-stochastic games. In this work, we present faster algorithms for 
mean-payoff parity games.

{
\begin{table}
\begin{center}
\begin{adjustbox}{max width=\textwidth}
\begin{tabular}{|l|}
\multicolumn{1}{l}{\phantom{horizon}} \\
\hline
\ MP-B\"uchi \ \\
\hline
\ MP-coB\"uchi \  \\
\hline
\ MP-parity\ \\
\hline
\end{tabular}
\begin{tabular}{|c|c|}

\multicolumn{2}{c}{threshold problem}\\
\hline
 Previous & Our  \\
\hline
$\O(n^3 \cdot m \cdot W)$ & $\O(n \cdot m \cdot W)$ \\
\hline
$\O(n^3 \cdot m \cdot W)$ & $\O(n \cdot m \cdot W)$ \\
\hline
$\O(n^{d+1} \cdot m \cdot W)$ & $\O(n^{d-1} \cdot m \cdot W)$ \\
\hline
\end{tabular}

\begin{tabular}{|c|c|}
\multicolumn{2}{c}{value problem}\\
\hline
 Previous & Our  \\
\hline
$\O(n^4 \cdot m \cdot W)$ & $\O(n^2 \cdot m \cdot W \cdot \log(nW))$ \\
\hline
$\O(n^4 \cdot m \cdot W)$ & $\O(n^2 \cdot m \cdot W \cdot \log(nW))$ \\
\hline
$\O(n^{d+2} \cdot m \cdot W)$ & $\O(n^{d} \cdot m \cdot W \cdot \log(nW))$ \\
\hline
\end{tabular}
\end{adjustbox}
\caption{Algorithmic bounds for mean-payoff (MP) and parity objectives, and special cases: 
threshold problem  (left) and value problem (right).
\label{tab:complexity}}
\end{center}
\vspace{-1em}
\end{table}
}

\section{Preliminaries}

\noindent{\emph{Graphs.}
A graph $G = (V, E)$ consists of a finite set $V$ of vertices and a finite set of edges 
$E \subseteq V\times V$. 
Given a graph $G = (V,E)$ and a subset $U \subseteq V$ we denote by $G\restr U = (V',E')$ 
the subgraph of $G$ induced by $U$, i.e., $V' = U$, $E' = (U \times U) \cap E$.
For $v\in V$ we denote by $\In(v)$ (resp., $\Out(v)$) the set of incoming (resp., outgoing) vertices, i.e., 
$\In(v)=\{v' \mid (v',v) \in E \}$, and $\Out(v)=\{v' \mid (v,v') \in E \}$.

\smallskip\noindent{\em Game graphs.}
A game graph $\Gamma = (V, E, \langle V_1, V_2 \rangle)$ is a graph whose 
vertex set is partitioned into $V_1$ and $V_2$, (i.e., $V = V_1 \cup V_2$ and $V_1 \cap V_2=\emptyset$). 
In a game graph every vertex $v \in V$ has a successor $v' \in V$, i.e.,
$\Out(v)\neq \emptyset$ for all $v \in V$. 
Given a game graph $\Gamma$ and a set $U$ such that for all vertices $u$ in $U$ we have 
$\Out(u) \cap U \neq \emptyset$, we denote by $\Gamma \restr U$ the subgame induced by $U$.

\smallskip\noindent{\emph{Plays.}
Given a game graph $\Gamma$ and a starting vertex $v_0$, 
the game proceeds in rounds.
In each round, if the current vertex belongs to player~1, then player~1 chooses
a successor vertex, and player~2 does likewise if the current vertex belongs to player~2.
The result is a {\em play} $\rho$ which is an infinite path from $v_0$, i.e., $\rho = v_0v_1\dots$, 
where every $(v_i, v_{i+1}) \in E$ for all $i\geq 0$.
We denote by $\Plays(\Gamma)$ the set of all plays of the game graph.

\smallskip\noindent{\em Strategies.}
Strategies are recipes to extend prefixes of plays by choosing the next vertex.
Formally, a strategy for player-1 is a function $\sigma_1:V^*\cdot V_1 \mapsto V$
such that $(v,\sigma_1(\rho \cdot v)) \in E$ for all $v \in V_1$ and all $\rho \in V^*$. 
We define strategies $\sigma_2$ for player~2 analogously.
We denote by $\Sigma_1$ and $\Sigma_2$ the set of all strategies for player~1 and player~2,
respectively.
Given  strategies $\sigma_1$ and $\sigma_2$ for player~1 and player~2, and a starting 
vertex $v_0$, there is a unique play $\rho=v_0 v_1 \ldots$ such that for all $i \geq 0$,
(a)~if $v_i \in V_1$ then $v_{i+1}=\sigma_1(v_0 \dots v_i)$; and 
(b)~if $v_i \in V_2$ then $v_{i+1}=\sigma_2(v_0 \dots v_i)$.
We denote the unique play as $\outc(v_0,\sigma_1,\sigma_2)$.
A strategy is {\em memoryless} if it is independent of the past and depends only on 
the current vertex, and hence can be defined as a function $\sigma_1: V_1 \mapsto V$
and $\sigma_2: V_2 \mapsto V$, respectively.

\smallskip\noindent{\emph{Objectives and parity objectives.}
An objective for a game graph $\Gamma$ is a subset of the possible plays, i.e.,
   $\phi \subseteq \Plays(\Gamma)$.
Given a play $\rho$ we denote by $\Inf(\rho)$ the set of vertices that appear 
infinitely often in $\rho$. 
A {\em parity} objective is defined with a priority function $p$ that maps every vertex to a non-negative
integer priority, and a play satisfies the parity objective for player~1 if the minimum priority 
vertex that appear infinitely often is even. 
Formally, the parity objective is $\Parity_\Gamma(p) = \{ \rho \in \Plays(\Gamma) \mid \min \{p(v) \mid 
v \in \Inf(\rho)\} \text{ is even} \}$. 
The Büchi and coBüchi objectives are special cases of parity objectives with two priorities only.
We have $p: V \mapsto \{0, 1\}$ for Büchi objectives and $p: V \mapsto \{1,2\}$ for the coBüchi objectives.

\smallskip\noindent{\em Payoff functions.} 
Consider a game graph $\Gamma$, and a weight function $w: E \mapsto \Z$ that maps every edge to
an integer.
The mean-payoff function maps every play to a real-number and is defined as follows:
For a play $\rho = v_0v_1\dots$ in $\Plays(\Gamma)$ we have $\MP(w,\rho) = \liminf\limits_{n\mapsto \infty} \frac{1}{n} \cdot \sum_{i=0}^{n-1} w(v_i,v_{i+1})$.
The mean-payoff parity function also maps every play to a real-number or $-\infty$ as follows:
if the parity objective is satisfied, then the value is the mean-payoff value, else it is $-\infty$.
Formally, for a play $\rho$, we have 
\[
\MPP_{\Gamma}(w,p,\rho)=\begin{cases}
\MP_{\Gamma}(w,\rho) & \text{ if } \rho \in \Parity_\Gamma(p); \\
- \infty & \text{ if }  \rho \not\in \Parity_\Gamma(p).
\end{cases}
\]

\smallskip\noindent{\em Threshold mean-payoff parity objectives.}
Given a threshold $\V \in \Q$, the {\em threshold mean-payoff objective} 
$\MeanPayoff_\Gamma(\V) = \{ \rho \in \Plays(\Gamma) \mid \MP(\rho) \geq \V \}$ 
requires that the mean-payoff value is at least $\V$.
The {\em threshold mean-payoff parity objective} is a conjunction of a parity 
objective and a threshold mean-payoff objective, i.e., $\Parity_\Gamma(p) \cap \MeanPayoff_\Gamma(\V)$.

\smallskip\noindent{\emph{Winning strategies.}
Given an objective (such as parity, threshold mean-payoff, or threshold mean-payoff parity) $\phi$,
a vertex $v$ is winning for player~1, if there is a strategy $\sigma_1$ such that for all
strategies $\sigma_2$ of player~2, the play $\outc(v,\sigma_1,\sigma_2) \in \phi$ (i.e., the play
satisfies the objective).
We denote by $W_1(\phi)$ the set of winning vertices (or the winning region) for player~1 for 
the objective $\phi$.
The notation $W_2(\overline{\phi})$ for complementary objectives $\overline{\phi}$ 
for player~2 is similar.

\smallskip\noindent{\em Value functions.}
Given a payoff function $f$ (such as the mean-payoff function, or the mean-payoff parity function),
the value for player~1 is the maximal payoff that she can guarantee against all strategies of 
player~2.
Formally, 
\[
\val_{\Gamma}(f)(v)= \sup_{\sigma_1 \in \Sigma_1} \inf_{\sigma_2 \in \Sigma_2} f(\outc(v,\sigma_1,\sigma_2)). 
\]

\smallskip\noindent{\emph{Attractors.}} 
The player-1 attractor $\attr_1(S)$ of a given set $S \subseteq V$ is the set of vertices
from which player-1 can force to reach a vertex in $S$. It is defined as the limit 
of the sequence $A_0 = S; 
A_{i+1} = A_i \cup \{ v \in V_1 \mid \Out(v) \cap A_i \neq \emptyset \} \cup \{v \in V_2 \mid \Out(v) \subseteq A_i\}$ for all $i \geq 0$. 
Th Player-2 attractor $\attr_2(S)$ is defined analogously exchanging the roles of player~1 and player~2.
The complement of an attractor induces a game graph, as in the complement every vertex has an outgoing
edge in the complement set.

\smallskip\noindent{\em Relevant parameters.}
In this work we will consider computing the winning region for threshold mean-payoff parity objectives,
and the value function for mean-payoff parity objectives.
We will consider the following relevant parameters:
$n$ denotes the number of vertices, $m$ denotes the number of edges, $d$ denotes the number of 
priorities of the parity function $p$, and $W$ is the maximum absolute value of the weight function $w$.

\section{Decremental Algorithm for Threshold Mean-Payoff Games}\label{sec:decr}
In this section we present a decremental algorithm for threshold mean-payoff games
that supports deleting a sequence of sets of vertices along with their player-2
attractors.
The overall running time of the algorithm is $\O(n \cdot m \cdot W)$.

\smallskip\noindent{\em Key idea.} 
A static algorithm based on the notion of progress measure for mean-payoff games 
was presented in~\cite{brim2011}.
We show that the progress measure is monotonic wrt to the deletion of vertices and their
player-2 attractors.
We use an amortized analysis to obtain the running time of our algorithm.

\smallskip\noindent{\em Mean-payoff progress measure.} 
Let $\Gamma$ be a mean-payoff game with threshold $\V$.
Progress measure is a function $f$ which maps every vertex in $\Gamma$ to an element of the set 
$C_\Gamma = \{i \in \N \mid i \leq nW \} \cup \{\top\}$, i.e., $f: V \mapsto C_\Gamma$. 
Let $(\preceq,C_\Gamma)$ be a total order, where $x \preceq y$ for $x,y \in C_\Gamma$
holds iff $x \leq y \leq nW$ or $y = \top$. We define the operation $\ominus:
C_\Gamma \times \Z \mapsto C_\Gamma$ for all $a \in C_\Gamma$ and $b \in \Z$ as
follows:
\[ 
a \ominus b = \begin{cases} 
\max(0,a-b) & \text{if $a \neq \top$ and $a-b \leq nW$,}\\
\top        & \text{otherwise.}
\end{cases}
\]
A player-1 vertex $v$ is \emph{consistent} if $f(v) \succeq f(v') \ominus
w(v,v')$ for \emph{any} $v' \in \Out(v)$. A player-2 vertex $v$ is \emph{consistent} 
if $f(v) \succeq f(v') \ominus w(v,v')$ for \emph{all} $v' \in \Out(v)$.
Let $v \in V$ then $\lift(\cdot,v) : [V \mapsto C_\Gamma] \mapsto [V \mapsto C_\Gamma]$ is defined by
$\lift(f,v) = g$ where:
\[ g(u) = 
\begin{cases}
f(u) & \text{ if $u \neq v$},\\
\min\{f(v') \ominus w(v,v') \mid (v,v') \in E \} & \text{ if $u = v \in V_1$},\\
\max\{f(v') \ominus w(v,v') \mid (v,v') \in E \} & \text{ if $u = v \in V_2$}.
\end{cases}
\]

\smallskip\noindent{\em Static Algorithm}
The static algorithm in~\cite{brim2011} is an iterative algorithm which
maintains and returns a progress measure $f$ and a list $L$ of vertices which are not consistent. 
The initial progress measure of every vertex is set to zero.
Also, $w(e)$ is set to $w(e) - \V$ for all edges $e$ in $E$.
The list $L$ is initialized with the vertices which are not consistent considering the 
initial progress measure. 
Then the following steps are executed in a while-loop:
\begin{enumerate}
\item Take out a vertex $v$ of $L$.
\item Perform the $\lift$-operation on the vertex, i.e., $f \leftarrow \lift(f,v)$. 
\item If a vertex $v'$ in $\In(v)$ is not consistent, put $v'$ into $L$.
\item If $L$ is empty, return $f$ else proceed to the next iteration.
\end{enumerate}
If every vertex is consistent, i.e., the list $L$ is empty, 
the winning region of player~1 is the set of vertices which are not set
to $\top$ in $f$, i.e., $W_1(\V) = \{ v \in V \mid f(v) \neq \top \}$.

\smallskip\noindent{\em Decremental input/output.} 
Let $\Gamma$ be a mean-payoff game with threshold $\V$. 
The \emph{input} to the decremental algorithm is a sequence of sets $A_1, A_2, \ldots,A_k$,
such that each $A_i$ is a player-2 attractor of a set $X_i$ in the game 
$\Gamma_i=\Gamma \restr (V \setminus \bigcup_{j<i} A_j)$.
The \emph{output requirement} is the player-1 winning set after the deletion of
$\bigcup_{j<i}A_j$ for $i=1,\dots,k$,
i.e., the output requirement is the sequence $Z_1,Z_2, \ldots,Z_k$, where 
$Z_i=W_1(\phi)$ in $\Gamma_i=\Gamma \restr (V \setminus \bigcup_{j<i} A_j)$,
where $\phi=\MeanPayoff_{\Gamma_i}(\V)$ is the threshold mean-payoff objective.
In other words, we repeatedly delete a vertex set $X_i$ along with its player-2 attractor $A_i$ 
from the current game graph $\Gamma_i$, and require the winning set for player~1 as an output after
each deletion.

\smallskip\noindent{\em Decremental algorithm.} 
We maintain a progress measure $f_i$, $1 \leq i \leq k$, during the whole sequence of deletions.
The initial progress measure $f_1$ for the mean-payoff game $\Gamma$ with threshold 
mean-payoff objective $\phi$ is calculated using the static algorithm.
For all edges $e$ in $E$, we set $w(e) = w(e) - \V$.
In iteration $i$ with 
input $A_i$, in the game $\Gamma_i$ with its corresponding vertex set $V_i$ the following steps are executed:
\begin{enumerate}
    \item If a vertex in the set $\{v \in V_i\setminus A_i \mid \exists v': \  v' \in \Out(v) \land v' \in A_i \}$
    is not consistent in $f_i$ without the set $A_i$, put it in a list $L_i$.
    \item Delete the set $A_i$ from $\Gamma_i$ to receive $\Gamma_{i+1}$ (and thus $V_{i+1}$).
    \item Execute steps (1)-(4) of the  above described iterative 
	algorithm from~\cite{brim2011} initialized with
    $\Gamma_{i+1}$, $L_i$ and $f_{i}$ restricted to the vertices in $V_{i+1}$.
    \item Finally the winning region of player~1 can be extracted from the obtained
	progress measure $f_{i+1}$, i.e., $W_1(\phi) = \{ v \in V_{i+1} \mid f(v) \neq \top \}$.
\end{enumerate}

\smallskip\noindent\emph{Correctness.}
Let $\Gamma$ be a game graph, $\phi$ a threshold objective and $A_1, A_2 ,\dots, A_k$ a sequence of sets, 
such that each $A_i$ is a player-2 attractor in the game 
$\Gamma_i = \Gamma \restr (V \setminus \bigcup_{j<i} A_j)$. To show the correctness of the decremental
algorithm we need to show that the condition that the list $L$ contains all vertices which are not consistent
is an invariant of the decremental algorithm at line~3. This property was proved for the static algorithm in~\cite{brim2011}.
\begin{lemma}
The condition that $L_i$ contains all vertices which are not consistent with the
progress measure $f_i$ restricted to $V_{i+1}$ in $\Gamma_{i+1}$ is an invariant of the
static algorithm called in step~3 of the decremental algorithm for $1 \leq i
\leq k-1$. 
\end{lemma}
\begin{proof}
The fact that the static algorithm correctly returns a progress measure with only 
consistent vertices when the invariant holds was shown in~\cite{brim2011}. It
was also shown in~\cite{brim2011} that the invariant is maintained in the loop.
It remains to show that the condition holds when we call the static algorithm at
step~3.
For the base case, let $i = 1$. In the initial progress measure $f_1$ and the initial game graph $\Gamma_1$, 
every vertex is consistent. By the definition of a player-2 attractor, deleting the set $A_1$ potentially
removes edges $(v,v')$ where $v$ is a player-1 vertex in $V\setminus A_1$
and $v'$ is in $A_1$. (Note that $v$ cannot be a player-2 vertex.) 
All of the vertices not consistent anymore are added to $L_i$ in step~1 of the decremental algorithm. For the inductive step let $i = j$. By induction hypothesis,
all vertices which were not consistent with the progress measures $f_{h-1}$ restricted to $V_{h}$ for 
$2\leq h \leq j$ were added to the corresponding lists. Thus by the correctness of the static algorithm, 
it correctly computes 
the new progress measure $f_{h}$ for the game graph $\Gamma_{h}$ where every vertex is consistent. 
Thus also every vertex in the progress measure $f_j$ restricted to $V_{j}$ is consistent. 
Again the player-2 attractor is removed and vertices which are not consistent with progress measure $f_j$ 
restricted to $V_{j+1}$ are put into $L_j$ by step~1 of the algorithm. 
\end{proof}
Thus we proved that the static algorithm always correctly updates to the new
progress measure in each iteration. 
The winning region of player-1 is obtained by the returned progress measure (step~4).
The decremental algorithm thus correctly computes the sequence 
$Z_1,Z_2, \dots Z_k$, where $Z_i = W_1(\phi)$ in $\Gamma_i$.

\smallskip\noindent\emph{Running Time.}
The calculation of the initial progress measure for the mean-payoff game $\Gamma$ with
threshold $\V$ is in time $\O(n \cdot m \cdot W)$.
The vertices which are not consistent anymore after the deletion of $A_i$ 
can be found in time $\O(m)$ (step~1). As at most $n$ such sets $A_i$ exist, the running time is
$\O(mn)$. In step~3 the static algorithm is executed with
our current progress measure $f_i$:
Every time a vertex $v$ is picked from 
the list $L_i$ it costs $\O(|\Out(v) + \In(v)|)$ time to
use $\lift$ on it and to look for vertices in $\In(v)$ which are not consistent anymore
(steps~1-3 in the static algorithm).
This cost is charged to its incident edges. Note that deleting a set of vertices and 
their corresponding player-2 attractor 
will only potentially \emph{increase} the progress measure of some player-1 vertices.
As we can increase the progress measure of every vertex only $nW$ times 
before it is set to $\top$ where it is always consistent, 
we get the desired bound of $\O(m\cdot n \cdot W)$.

\noindent Thus our decremental algorithm for threshold mean-payoff games works as desired and we obtain 
the following result:

\begin{theorem}~\label{thm:decr}
Given a game graph $\Gamma$, a threshold mean-payoff objective $\phi$ and a
sequence of sets $A_1, A_2, \dots, A_k$ such that each $A_i$ is a player-2 
attractor of a set $X_i$ in the game $\Gamma_i=\Gamma \restr (V \setminus \bigcup_{j<i} A_j)$, 
the sequence $Z_1,Z_2,\dots, Z_k$, where $Z_i = W_1(\phi)$ in $\Gamma_i$ can be computed in
$\O(n \cdot m \cdot W)$ time.
\end{theorem}

\begin{remark}\label{rem:nodeladd}
Note that the running time analysis of our decremental algorithm crucially
depends on the monotonicity property of the progress measure.
If edges are both added and deleted, then the monotonicity property does not hold.
Hence obtaining a fully dynamic algorithm that supports both addition/deletion of 
vertices/edges with running time $\O(n \cdot m \cdot W)$ is an interesting open problem.
However, we will show that for solving mean-payoff parity games, the decremental 
algorithm plays a crucial part.
\end{remark}

\section{Threshold Mean-Payoff Parity Games}
In this section we present algorithms for threshold mean-payoff parity games.
Our most interesting contributions are for the base case of 
mean-payoff Büchi- and mean-payoff coBüchi objectives, and the general case follows
a standard recursive argument. 

\subsection{Threshold Mean-Payoff Büchi Games}
In this section we consider threshold mean-payoff B\"uchi games.

\smallskip\noindent{\em Algorithm for threshold mean-payoff B\"uchi games.}
The basic algorithm is an iterative algorithm that deletes player-2 attractors.
The algorithm proceeds in iterations. In iteration $i$, let $D_i$ be the set of 
vertices already deleted. Consider the subgame $\Gamma_i=\Gamma \restr (V\setminus D_i)$.
Then the following steps are executed:
\begin{enumerate}
\item Let $V^i= V \setminus D_i$ and $B_i$ denote the set of B\"uchi vertices (or vertices with 
priority~0) in $\Gamma_i$.
Compute $Y_i=\attr_1(B_i)$ the player-1 attractor to $B_i$ in $\Gamma_i$.

\item Let $X_i=V^{i} \setminus Y_i$. If $X_i$ is non-empty, remove $A_i=\attr_2(X_i)$  from the game graph,
and proceed to the next iteration.

\item Else $V^{i}=Y_i$. Let $U_i=W_1(\phi)$ in $\Gamma_i$, where $\phi=\MeanPayoff(\V)$, be the winning 
region for the threshold mean-payoff objective in $\Gamma_i$. 
Let $X_i=V^{i} \setminus U_i$. 
If $X_i$ is non-empty, remove $A_i=\attr_2(X_i)$  from the game graph,
and proceed to the next iteration.
If $X_i$ is empty, then the algorithm stops and all the remaining vertices are winning for player~1
for the threshold mean-payoff B\"uchi objective. 
\end{enumerate}

\smallskip\noindent\emph{Correctness.}
Since the correctness argument has been used before~\cite{CJH05}, we only present a brief sketch:
The basic correctness argument is to show that all vertices removed over all iterations 
do not belong to the winning set for player~1. 
In the end, for the remaining vertices, player~1 can ensure to reach the B\"uchi vertices,
and ensures the threshold mean-payoff objectives.
A strategy that plays for the threshold mean-payoff objectives longer and longer, and in between
visits the B\"uchi vertices, ensures that the threshold mean-payoff B\"uchi objective is satisfied.

\smallskip\noindent{\em Running time analysis.}
We observe that the total running time to compute all attractors is at most $\O(n \cdot m)$,
since the algorithm runs for $\O(n)$ iterations and each attractor computation is linear time.
In step~3, the algorithm needs to compute the winning region for threshold mean-payoff 
objective.
The algorithm always removes a set $X_i$ and its player-2 attractor $A_i$, and requires
the winning set for player~1. 
Thus we can use the decremental algorithm from Section~\ref{sec:decr}, which precisely supports
these operations.  
Hence using Theorem~\ref{thm:decr} in the algorithm for threshold mean-payoff B\"uchi games, 
we obtain the following result.

\begin{theorem}\label{thm:buchi}
Given a game graph $\Gamma$ and a threshold mean-payoff B\"uchi objective $\phi$,
the winning set $W_1(\phi)$ can be computed in $\O(m\cdot n \cdot W)$ time.
\end{theorem}

\subsection{Threshold Mean-Payoff coBüchi Games}
In this section we will present an $\O(n \cdot m \cdot W)$-time algorithm for threshold
mean-payoff coB\"uchi games. 
We start with the description of the basic algorithm for threshold mean-payoff 
coB\"uchi games.

\smallskip\noindent{\em Algorithm for threshold mean-payoff coB\"uchi games.}
The basic algorithm is an iterative algorithm that deletes player-1 attractors.
The algorithm proceeds in iteration. In iteration $i$, let $D_i$ be the set of 
vertices already deleted. Consider the subgame $\Gamma_i=\Gamma \restr (V\setminus D_i)$.
Then the following steps are executed:
\begin{enumerate}
\item Let $V^i= V \setminus D_i$ and $C_i$ denote the set of coB\"uchi vertices (or vertices with 
priority~1) in $\Gamma_i$.
Compute $Y_i=\attr_2(C_i)$ the player-2 attractor to $C_i$ in $\Gamma_i$.

\item Let $X_i=V^{i} \setminus Y_i$. Consider the subgame $\widehat{\Gamma}_i=\Gamma_i \restr X_i$.
Compute the winning region $Z_i=W_1(\phi)$ for player~1 in $\widehat{\Gamma}_i$, where $\phi=\MeanPayoff(\V)$
is the threshold mean-payoff objective.

\item If $Z_i$ is non-empty, remove $\attr_1(Z_i)$ from $\Gamma_i$, and proceed to the next
iteration. Else if $Z_i$ is empty, then all remaining vertices are winning for player~2.
\end{enumerate}

\smallskip\noindent{\em Correctness argument.}
Consider the subgame $\Gamma_i$. In each subgame $\widehat{\Gamma}_i$ of $\Gamma_i$ 
all edges of player~2 are intact, since it is obtained after removing a player-2 attractor $Y_i$.
Moreover, there is no priority-1 vertex in $\widehat{\Gamma}_i$. 
Hence ensuring the threshold mean-payoff objective in $\widehat{\Gamma}_i$ for player~1 ensures
satisfying the threshold mean-payoff coB\"uchi objective. Hence the set $Z_i$ and its player-1 
attractor belongs to the winning set of player~1 and can be removed.
Thus all vertices removed are part of the winning region for player~1.
Upon termination, in $\widehat{\Gamma}_i$, player~1 cannot satisfy the threshold mean-payoff 
condition from any vertex.
Consider a player-2 strategy, where in $\widehat{\Gamma}_i$ player~2 falsifies the threshold 
mean-payoff condition, and in $Y_i$ plays an attractor strategy to reach $C_i$ (priority-1 vertices).
Given such a strategy, either 
(a)~$Y_i$ is visited infinitely often, and then the coB\"uchi objective is violated; or
(b)~from some point on the play stays in $\widehat{\Gamma}_i$ forever, and then the threshold
mean-payoff objective is violated.
This shows the correctness of the algorithm. 
However, the running time of this algorithm is not $\O(n \cdot m \cdot W)$. 
We now present the key ideas to obtain an $\O(n \cdot m \cdot W)$-time algorithm.

\smallskip\noindent{\em First intuition.}
Our first intuition is as follows. 
In step~2 of the above algorithm, instead of obtaining the whole winning region $W_1(\phi)$
in $\widehat{\Gamma}_i$ it suffices to identify a subset $X_i$ of the winning region
(if it is non-empty) and remove its player-1 attractor.
We call this the modified algorithm for threshold mean-payoff coB\"uchi games.
We first describe why we cannot use the decremental approach in the following remark.

\begin{remark}
Consider the subgames for which the threshold mean-payoff objective must be solved.
Consider Figure~\ref{fig:nodecr}. 
The first player-2 attractor removal induces subgame $\widehat{\Gamma}_1$.
After identifying a winning region $X_1$ of $\widehat{\Gamma}_1$ we remove its player-1 attractor $A_1$.
After removal of $A_1$, we consider the second player-2 attractor to the
priority-1 vertices.
The removal of this attractor induces $\widehat{\Gamma}_2$.
We observe comparing $\widehat{\Gamma}_1$ and $\widehat{\Gamma}_2$ that certain 
vertices are removed, whereas other vertices are added.
Thus the subgames to be solved for threshold mean-payoff objectives do not satisfy
the condition of decremental or incremental algorithms (see Remark~\ref{rem:nodeladd}).

\begin{figure}[H]
\centering
\begin{tikzpicture}
\node at (1.35,1.09) {$X_1$};
\node at (1.3,3) {$\widehat{\Gamma}_1$};
\node at (4.0,3.1) {$\widehat{\Gamma}_2$};
\node at (9.0,2.5) {1};
\draw [-,decorate,decoration={snake,amplitude=.4mm,segment length=2mm,post
    length=1mm}, line width=0.4mm]
    (3.5,0) -- (3.5,5);
\draw [-,decorate,decoration={snake,amplitude=.4mm,segment length=2mm,post
    length=1mm}, line width=0.4mm]
    (4.5,5) -- (4.5,1.54);
\draw (0,0) rectangle (10,5);
\draw (0,0) rectangle (3,2);
\draw[draw =none, pattern=del lines] (0,0) -- (3.5,0) -- (3.5,1.85) --
    (3,2) -- (0,2) -- (0,0);

\draw (10.5,1.5) rectangle (13,4);
\draw [pattern=added lines] (10.7,1.7) rectangle (11.7,2.6);
\draw [pattern= del lines] (10.7,2.9) rectangle (11.7,3.8);
\node at (12.4,3.3) {deleted};
\node at (12.4,2.15) {added};
\draw[draw=none, pattern=added lines] (3.5,1.86) -- (4.5,1.54) -- (4.5,5) -- (3.5,5)
    -- (3.5,1.86);
\draw (8,0) -- (8,5);
\draw [-,decorate,decoration={snake,amplitude=.4mm,segment length=2mm,post length=1mm}]
    (3,2) -- (10,0);
\node at (6.5,0.5) {$A_1$};
\node at (5.0,4.5) {\textbf{player-2 attr.(1)}};
\node at (6.0,2.5) {\textbf{player-2 attr.(2)}};

\end{tikzpicture}
\caption{Pictorial illustration of threshold mean-payoff coBüchi games. The
    subgames $\widehat{\Gamma}_1$ and $\widehat{\Gamma}_2$ are shown. We observe that $\widehat{\Gamma}_2$ is
        obtained both by addition and deletion of game parts to $\widehat{\Gamma}_1$.} \label{fig:nodecr}
\end{figure}
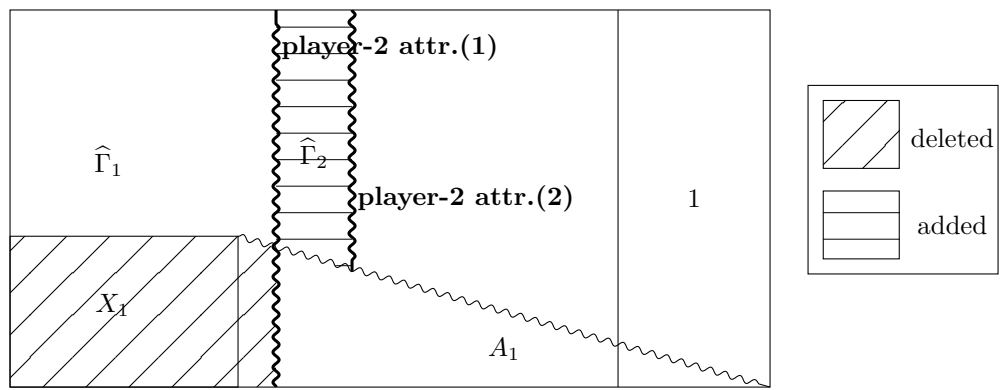
\end{remark}

\smallskip\noindent{\em Second intuition.} 
While we cannot use the decremental algorithm, we can solve the problem in 
$\O(n \cdot m \cdot W)$ time, if we have a modified static algorithm for threshold mean-payoff 
games, with the following property:
(a)~it identifies a subset of the winning region $X$ for player~1, if the winning region
is non-empty, in time $\O(|X| \cdot m \cdot W)$;
(b)~if the winning region is empty, it returns the empty set, and then it takes
time $\O(n \cdot m \cdot W)$.
With such an algorithm we analyze the running time of the above modified algorithm for threshold 
mean-payoff coB\"uchi games.
The total time required for all attractor computations is again $\O(n \cdot m)$. 
Otherwise, we use the modified static algorithm to remove vertices of player-1 and to remove set of size 
$X$ we take $\O(|X| \cdot m \cdot W)$ time, and thus we can charge each vertex
$\O(m \cdot W)$ time. 
Hence the total time required is $\O(n \cdot m \cdot W)$.
In the rest of the section we present this modified static algorithm for threshold mean-payoff games.

\smallskip\noindent\emph{Problem Statement.}
\begin{framed}
\begin{tabular}{ l l}
\textbf{Input:}    & Mean-payoff game $\Gamma$ with threshold $\V$.\\
\textbf{Question:} & If $W_1(\MeanPayoff(\V))$ is non-empty, return a nonempty set\\
&$X \subseteq W_1(\MeanPayoff(\V))$ in time $\O(|X| \cdot m \cdot W)$,\\
&else return $\emptyset$ in time $\O(n \cdot m \cdot W)$.
\end{tabular}
\end{framed}

\smallskip\noindent{\em Modified static algorithm for threshold mean-payoff games.}
The basic algorithm for threshold mean-payoff games computes a progress measure, with a 
defined top element value $\top$. 
If the progress measure has the value $\top$ for a vertex, then the vertex is 
declared as winning for player~2. 
With value $\top=n\cdot W$, the correct winning region for both players can be 
identified. 
Moreover, for a given value $\alpha$ for $\top$, the progress measure algorithm 
requires $\O(\alpha \cdot m)$ time.
Our modified static algorithm is based on the following idea:
\begin{enumerate}
\item Consider a value $\alpha \leq n \cdot W$ for the top element. 
With this reduced value for the top element, if a winning region is identified
for player~1, then it is a subset of the whole winning region for player~1.
\item We will iteratively double the value for the top element. 
\end{enumerate}
Given the above ideas our algorithm is an iterative algorithm defined as follows:
Initialize top value $\top_0=W$. The $i$-th iteration is as follows:
\begin{enumerate}
\item Run the progress measure algorithm with top value $\top_i$.

\item If a winning region $X$ for player is identified, return $X$.

\item Else $\top_{i+1}=2 \cdot \top_i$ (i.e., the top value is doubled).

\item If $\top_{i+1} \geq 2 \cdot n \cdot W$, stop the algorithm and return
$\emptyset$, else proceed to the next iteration.
\end{enumerate}
Details can be found in Appendix~\ref{app:winningset}.

\smallskip\noindent{\em Correctness and running time analysis.}
The key steps of the correctness argument and the running time analysis are as follows:
\begin{enumerate}

\item The above algorithm is correct, since if it returns a set $X$ then it is a subset of 
the winning set for player~1. 

\item If the algorithm returns a winning set with top value $\alpha$, then the total 
running time till this iteration is $m \cdot (\alpha + \alpha/2 + \alpha/4 + \cdots)$, because the 
progress with top value $\alpha$ requires time $\O(\alpha \cdot m)$. 
Hence the total running time if a set $X$ is returned with top value $\alpha$ is
$\O(\alpha \cdot m)$.

\item Let $Z$ be a set of vertices such that no player-2 vertex in $Z$ has an edge out of $Z$,
and the whole subgame $\Gamma \restr Z$ is winning for player~1. Then a winning strategy in 
$Z$ ensures that a progress measure with top value $|Z| \cdot W$ would identify the set 
$Z$ as a winning set.

\item From above it follows that if the winning set $X$ is identified at top value $\alpha$, 
but no winning set was identified with top value $\alpha/2$, then the size of the winning set 
is at least $\alpha/(2W)$.

\item It follows from above that if a set $X$ is identified, then the total running time 
to obtain set $X$ is $\O(|X| \cdot m \cdot W)$.

\item Moreover, the total running time of the algorithm when no set $X$ is identified is in 
$\O(n \cdot m \cdot W)$, and in this case, the winning region is empty.

\end{enumerate}
Thus we solved the modified static algorithm for threshold mean-payoff games as desired
and obtain the following result.

\begin{theorem}\label{thm:winningset}
Given a mean-payoff game $\Gamma$ and a threshold $\V$, let
$Z=W_1(\MeanPayoff(\V))$. 
If $Z\neq \emptyset$, then a non-empty set $X \subseteq Z$ can be computed in
time $\O(|X| \cdot m \cdot W)$, 
else an empty set is returned if $Z=\emptyset$, which takes time $\O(n\cdot m\cdot W)$.
\end{theorem}

Using the above algorithm to compute the winning set for player~1 in the
subgames, we obtain an algorithm for threshold mean-payoff coBüchi
games in time $\O(n\cdot m\cdot W)$. Details can be found in
Appendix~\ref{app:mpcobuchi}.

\begin{theorem}~\label{thm:cobuchi}
Given a game graph $\Gamma$ and a threshold mean-payoff coB\"uchi objective $\phi$,
the winning set $W_1(\phi)$ can be computed in $\O(n\cdot m \cdot W)$ time.
\end{theorem}

\subsection{Threshold Mean-Payoff Parity Games}
The algorithm for threshold mean-payoff parity games is the standard recursive
algorithm~\cite{CJH05}
(classical parity game-style algorithm) that generalizes the B\"uchi and coB\"uchi cases
(which are the base cases).
The running time recurrence is as follows: 
$T(n,d,m,w) = n (T(n,d-1,m) + \O(m)) + \O(nmW)$.
Using our approach we obtain the following result (details in Appendix).

\begin{theorem}~\label{thm:tmpp}
Given a game graph $\Gamma$ and a threshold mean-payoff parity objective $\phi$,
the winning set $W_1(\phi)$ can be computed in $\O(n^{d-1} \cdot m \cdot W)$ time.
\end{theorem}

\section{Optimal Values for Mean-payoff Parity Games}
In this section we present an algorithm which computes the value function for 
mean-payoff parity games. For mean-payoff games a dichotomic search approach was 
presented in~\cite{brim2011}.
We show that such an approach can be generalized to mean-payoff parity games.

\smallskip\noindent\emph{Range of Values for the Dichotomic Search.}
To describe the algorithm we recall a lemma about the possible range of optimal 
values of a mean-payoff parity game.
The lemma is an easy consequence of the characterization of~\cite{CJH05} 
that the mean-payoff parity value coincide with the mean-payoff value,
and the possible range of value for mean-payoff games.

\begin{lemma}[\cite{CJH05,Ehrenfeucht1979,Lifshits2007}]\label{lem:optval}
Let $\Gamma$ be a mean-payoff parity game. For each vertex $v \in V$, the optimal value  
$\val_{\Gamma}(\MPP)(v)$ is a rational number $\frac{y}{z}$ such that $1 \leq z \leq n$ 
and $|y| \leq z \cdot W$.
\end{lemma}

By Lemma~\ref{lem:optval} the value of each vertex $v \in V$, is contained in the following set of rationals 
\begin{align*}
S^\Gamma = \bigg\{ \frac{y}{z}\  \bigg | \ y,z \in \Z, 1 \leq z \leq n \land -z \cdot W \leq y \leq z \cdot W \bigg\}.
\end{align*}

\begin{definition}
Let $\Gamma$ be a mean-payoff parity game.
We denote the set of vertices $v \in V$ such that 
$\val_{\Gamma}(\MPP)(v) \circ \mu$ where $\circ \in \{<,\leq,=,\geq,>\}$
with $V^{\circ \mu}_\Gamma$.
\end{definition}

\smallskip\noindent\emph{Key Observation.}
Let $\Gamma = (V,E,\langle V_1, V_2 \rangle,w,p)$ be a mean-payoff parity game.
Let $\mu \in [-W,W]$. The sets $V^{>\mu}_\Gamma,V^{=\mu}_\Gamma$ and $V^{<\mu}_\Gamma$ 
can be computed using any algorithm for threshold mean-payoff parity games 
twice (for example using Theorem~\ref{thm:tmpp}). 
To calculate $V^{\geq \mu}_\Gamma$ and $V^{< \mu}_\Gamma$ use the algorithm on $\Gamma$ with the
mean-payoff parity objective $\phi = \Parity_\Gamma(p) \cap \MeanPayoff_\Gamma(\mu)$.
Consider $\Gamma' = (V,E,\langle V_2, V_1 \rangle,w',p)$, where $w'(e) = -w(e)$ for all edges $e \in E$ 
and player-1 and player-2 vertices are swapped.
To calculate $V^{\leq \mu}_\Gamma$ and $V^{> \mu}_\Gamma$ use the algorithm on $\Gamma'$ with 
mean-payoff parity objective $\phi = \Parity_{\Gamma'}(p) \cap \MeanPayoff_{\Gamma'}(-\mu)$. 
Given the sets $V^{\leq \mu}_\Gamma$, $V^{> \mu}_\Gamma, V^{\geq \mu}_\Gamma$ and $V^{< \mu}_\Gamma$
we can extract the sets $V^{>\mu}_\Gamma,V^{=\mu}_\Gamma$ and $V^{<\mu}_\Gamma$.\\
All values $\mu'$ in  $S^\Gamma$ are of the form $\frac{y}{z}$. For those values
we can determine whether $v \in V_{\Gamma}^{\geq \mu'}$ by applying 
the algorithm for threshold mean-payoff parity games on $\Gamma' = (V,E,\langle V_2, V_1 \rangle,w',p)$ where 
$w'(e) = w(e) \cdot z$ for all $e \in E$ with the mean-payoff parity objectives $\phi = \Parity_\Gamma(p) \cap
\MeanPayoff_\Gamma(y)$. Note that in the worst case, the weight function $w'$ of
$\Gamma'$ is in $\O(nW)$.

\smallskip\noindent\emph{Dichotomic Search.}
Let $\Gamma$ be a mean-payoff parity game. The dichotomic search algorithm is recursive algorithm
initialized with $\Gamma_0 = \Gamma$ and $S_0 = S^\Gamma$. In recursive call $i$ the following steps are
executed:
\begin{enumerate}
\item Let $r_i =  \min(S_i)$ and $s_i = max(S_i)$.
\item Determine $a_1$, the largest element in $S_i$ less than or equal to $\frac{r_i + s_i}{2}$ and  
$a_2$, the smallest element in $S_i$ greater than or equal to $\frac{r_i + s_i}{2}$.
\item Determine the partitions $V_{\Gamma_i}^{<a_1}$, $V_{\Gamma_i}^{=a_1}$, 
$V_{\Gamma_i}^{=a_2}$, $V_{\Gamma_i}^{>a_2}$ using the key observation.
\item For all $v \in V_{\Gamma_i}^{=a_1}$ set the value to $a_1$,
for all $v \in V_{\Gamma_i}^{=a_2}$ set the value to $a_2$ and set the value to  $-\infty$
for all vertices $v$ which are not in any set calculated in step 3.
\item Recurse upon $\Gamma_i \restr V_{\Gamma_i}^{<a_1}$ and $\Gamma_i \restr V_{\Gamma_i}^{>a_2}$.
\end{enumerate}

\smallskip\noindent\emph{Correctness.}
Let $\Gamma$ be a mean-payoff parity game. We prove that the dichotomic search algorithm correctly calculates
$\val_\Gamma(\MPP)(v)$ for all $v \in V$.
The algorithm is initialized with $\Gamma$ and $S^\Gamma$. 
By Lemma~\ref{lem:optval} the values of the vertices 
$v \in V$ are in the set $S^\Gamma$. Because we perform a binary search over the set $S^\Gamma$ 
we can guarantee the
termination of the algorithm. Notice that we need to show that the values calculated in the 
subgames constructed in step~4 are identical to the values in the original game. Then correctness follows 
immediately by our key observation and because we perform a binary search over the set $S^\Gamma$.

\begin{lemma}\label{lem:subgameval}
Given a mean-payoff parity game $\Gamma$ and $\mu \in \Q$, 
let $\Gamma' = \Gamma \restr V^{>\mu}_\Gamma$ and 
$\Gamma'' = \Gamma \restr V^{<\mu}_\Gamma$. 
For all $v \in V^{> \mu}_\Gamma$, we have $\val_{\Gamma'}(\MPP)(v) = \val_{\Gamma}(\MPP)(v)$ and 
for all $v \in V^{< \mu}_\Gamma$, we have $\val_{\Gamma''}(\MPP)(v) = \val_{\Gamma}(\MPP)(v)$.
\end{lemma}

\begin{proof}
Let  $v \in V^{> \mu}_\Gamma$ be arbitrary.
We will prove $\val_{\Gamma'}(\MPP)(v) = \val_{\Gamma}(\MPP)(v)$ by showing the following two cases: 
\begin{itemize}
	\item $\val_{\Gamma'}(\MPP)(v) \leq \val_{\Gamma}(\MPP)(v)$: 
	Note that there can be no player-2 vertex in $V^{>\mu}_\Gamma$ with an edge to $V^{\leq \mu}_\Gamma$.
	Thus we cut away only edges of player-1 vertices in $\Gamma'$. 
	Consequently player-1 has less choices in $\Gamma'$ than in $\Gamma$ at each of her vertices.
	Thus $\val_{\Gamma'}(\MPP)(v) \leq \val_{\Gamma}(\MPP)(v)$ holds.

	\item $\val_{\Gamma'}(\MPP)(v) \geq \val_{\Gamma}(\MPP)(v)$:
	Let $\sigma_1$ be an optimal strategy for player~1 and let
	$\sigma_2$ be an optimal strategy for player~2 which both exist
	by~\cite{CJH05}.
	We will show that $\sigma_1$ produces plays with vertices in $V^{> \mu}_\Gamma$ only, 
	if it starts from $v$.
	For the sake of contradiction assume that a play $\rho = \outc(v,\sigma_1,\sigma_2)$ 
	contains a vertex $v^* \in  V^{\leq \mu}_\Gamma$.
	Notice that there are no player-2 vertices in $V^{>\mu}_\Gamma$ with edges to $V^{\leq \mu}_\Gamma$. 
	Thus $\sigma_1$ chose a successor vertex in $V^{\leq \mu}_\Gamma$. 
	But when $\rho$ ends up in $V^{\leq \mu}_\Gamma$ 
	the optimal player-2 strategy $\sigma_2$ can guarantee that $\MPP_\Gamma(w,p,\rho) \leq \mu$
	by the definition of $V^{\leq \mu}_\Gamma$. 
	There is a strategy to keep the value of the play starting at
	$v$ greater than $\mu$ by the definition of $V^{>\mu}_\Gamma$. Thus any play $\rho$ leading 
	to $V^{\leq \mu}_\Gamma$ by $\sigma_1$ is not optimal which is a contradiction to our assumption. 
	Consequently $\val_{\Gamma'}(\MPP)(v) \geq \val_{\Gamma}(\MPP)(v)$ follows.
\end{itemize}
The fact that for all $v \in V^{< \mu}_\Gamma$, we have
$\val_{\Gamma''}(\MPP)(v) = \val_{\Gamma}(\MPP)(v)$ follows by a symmetric argument.

%
%
%

\end{proof}

\smallskip\noindent\emph{Running Time.}
The running time of the dichotomic search is $\O(n \cdot \log(nW) \cdot \mathsf{TH})$
where $\mathsf{TH}$ is the running time of an algorithm for the threshold mean-payoff parity problem.
The additional factor $n$ comes from rescaling the weights of the mean-payoff parity game $\Gamma$ which is
described in the key observation.
The factor $\O(\log(nW))$ is from using binary search on $S$ as $|S| = \O(n^2 \cdot W)$.

\begin{theorem}
Given a game graph $\Gamma$ and an algorithm
that solves the threshold mean-payoff parity problem in $\O(\mathsf{TH})$, the
value function of $\Gamma$ can be computed in time
$\O(n \cdot \log(nW) \cdot \mathsf{TH} )$.
\end{theorem}

\noindent As a corollary of the above theorem and Theorem~\ref{thm:tmpp}, the value function 
for mean-payoff parity games can be computed in $\O(n^d \cdot m \cdot W \cdot\log(nW))$ 
time. 

\section{Conclusion}
In this paper we present faster algorithms for mean-payoff parity games. 
Our most interesting results are for mean-payoff Büchi and mean-payoff coBüchi games, 
which are the base cases. For threshold mean-payoff Büchi and mean-payoff coBüchi games, 
our bound $\O(n\cdot m \cdot W)$ matches the current best-known bound for mean-payoff games. 
For the value problem, we show the dichotomic search approach of~\cite{brim2011} 
for mean-payoff games can be generalized to mean-payoff parity games. 
This gives an additional multiplicative factor of $n \cdot \log(nW)$ as compared to the threshold problem. 
A recent work~\cite{comin2017value} shows that the value problem for mean-payoff objective can be 
solved with a multiplicative factor $n$ compared to the threshold objective 
(i.e., it shaves of the $\log$ factor). 
An interesting question is whether the approach of~\cite{comin2017value} can be generalized to mean-payoff parity games.

\subparagraph*{Acknowledgements}
The authors are partially supported by the Vienna Science and Technology Fund
(WWTF) grant ICT15-003.
For M. H. and A. S. the research leading to these results has received funding
from the European Research Council under the European Union’s Seventh Framework
Programme (FP/2007-2013) / ERC Grant Agreement no. 340506.
K.C. is supported by the Austrian Science Fund (FWF) NFN Grant No S11407-N23 (RiSE/SHiNE) and an ERC Start grant (279307: Graph Games).
\bibliography{mpg}
\appendix

\section{Details of the modified static algorithm for threshold mean-payoff games}\label{app:winningset}
The formal description of the winning set algorithm is as follows:\\
\begin{algorithm}[H]\label{alg:winfast}
\SetAlgoLined

\DontPrintSemicolon
\KwIn {A mean-payoff game $\Gamma = ((V, E, \langle V_1, V_2 \rangle), w)$ with mean-payoff objective $\phi$.}
\KwOut {A set of winning vertices $X \subseteq W_1(\phi)$ or $\emptyset$}
\Begin{
        i=1;\\
	\While{$i \leq n$}{
        Define $\ominus_i: C_\Gamma \times \Z \mapsto C_\Gamma$, where
            $C_\Gamma = \{ j \in \N \mid j \leq iW\} \cup \{\top\}$
            $$a \ominus_i b = \begin{cases}
        \max(0,a-b) & \text{if $a \neq \top$ and $a-b \leq iW$} \\
            \top & \text{otherwise}
        \end{cases}
        $$\\

		Use the static algorithm of~\cite{brim2011} 
		described in Section~\ref{sec:decr} 
		(replacing every occurrence of the original $\ominus$ with $\ominus_i$) on $\Gamma$ and $\phi$. \\
		The algorithm will return a progress measure $f$ where every vertex is consistent and we 
		can obtain the winning set: $X = \{ v \mid f(v) \neq \top \}$.\\
        \If{$X \neq \emptyset$}{
			\Return $X$;
		}
		$i = \min (i\cdot2,n)$;\\
	}
	\Return {$\emptyset;$}
}
\caption{Calculating a winning set of a game mean-payoff game $\Gamma$}
\end{algorithm}

\begin{lemma}[Correctness]\label{lem:corrfast}
Given a game graph $\Gamma$ and mean-payoff objectives\\ $\phi = \MeanPayoff(\V)$,
Algorithm~\ref{alg:winfast} returns a winning set $X \subseteq W_1(\phi)$ or
$\emptyset$ if no such set exists.
\end{lemma}
\begin{proof}
Let Algorithm~\ref{alg:winfast} return a progress measure $f$
for $\Gamma$ at line 5 and the set $X$ is not empty. 
By the correctness (shown in~\cite{brim2011}) 
of the static algorithm used in step~5, $X\subseteq W_1(\phi)$.
Assume now that Algorithm~\ref{alg:winfast} returns $\emptyset$.
Because $i$ will at some point be greater or equal to $n$,
note that the original static algorithm is then executed at line~5. Again by its correctness we get that
there are no winning vertices and thus $\emptyset$ is the correct result.
\end{proof}

\begin{lemma}[Running Time]
Algorithm~\ref{alg:winfast} returns a winning set $X$ in $\O(|X| \cdot m \cdot W)$
or $\emptyset$ in $\O(n \cdot m \cdot W)$.
\end{lemma}
\begin{proof}
If Algorithm~\ref{alg:winfast} terminates at line~8 in iteration $i \leq n$ returning $X$,
the total running time until this iteration is $m\cdot W \cdot ( i + i/2 + i/4 + \cdots)$ because using
the static algorithm with $\ominus_i$ requires time $\O(i \cdot m \cdot W)$. 
Thus, when a set $X$ is returned in iteration $i$, it requires time $\O(i \cdot m \cdot W)$.
Let $Z$ be a set of vertices such that no player-2 vertex in $Z$ has an edge out of $Z$,
and the whole subgame $\Gamma\restr Z$ is winning for player~1. Then a winning strategy in $Z$ ensures that
a progress measure with $i = |Z|$ would identify the set $Z$ as a winning set.
From our assumption that Algorithm~\ref{alg:winfast} terminates at $i \leq n$ we know that no 
winning set was identified when $i$ had value $i/2$. Thus the returned set $X$ had size greater than 
$i/2$. Therefore when a set $X$ is returned in iteration $i$, it requires time $\O(|X| \cdot m \cdot W)$. 
If Algorithm~\ref{alg:winfast} terminates at line~12 returning $\emptyset$ we have a runtime $\O(n \cdot m
\cdot W)$ as $i$ was $n$ in the last iteration.
\end{proof}

The last two lemmas yield Theorem~\ref{thm:winningset}.

\section{Details of the Mean-Payoff coBüchi Algorithm}\label{app:mpcobuchi}
Algorithm~\ref{alg:cobuchi} is the new algorithm for threshold mean-payoff coBüchi games, whose
correctness is the same as the correctness of the basic algorithm for threshold
mean-payoff coBüchi games. Using Theorem~\ref{thm:winningset} for line~6 we obtain that the running time is
$\O(n\cdot m\cdot W)$ and hence obtain Theorem~\ref{thm:cobuchi}.

\begin{algorithm}
\caption{SolveMeanPayoffcoBüchi}\label{alg:cobuchi}
\KwIn {A mean-payoff coBüchi game $\Gamma = ((V,E,\langle V_1, V_2 \rangle,w, p)$ and objectives $\phi =
		\MeanPayoff(\V) \cap \Parity(p)$}
\KwOut {The winning region of player~1 in $\Gamma$, i.e., $W_1(\phi)$.}
\Begin{
$i \leftarrow 0$, $B_i \leftarrow V$\;
\Repeat{$B_i = B_{i+1}$}{
    $Y_i \leftarrow \attr_2(B_i \cap p^{-1}(1))$\;
    $X_i \leftarrow B_i \setminus Y_i$\;
    Compute a subset of the winning region $Z_i \subseteq W_1(\phi)$ for player~1 in $\Gamma\restr
    X_i$ where $\phi = \MeanPayoff(\V)$ is the threshold mean-payoff
    objective.\;
    $B_{i+1} \leftarrow B_i \setminus \attr_1(Z_i)$\;
    $i \leftarrow i+1$
}
\Return $V \setminus B_i$
}
\end{algorithm}

\section{Details of the Mean-Payoff Parity Algorithm}\label{app:mpp}
We recall the algorithm for mean-payoff parity games~\cite{CJH05}, 
and present the relevant details for the sake of being self-contained. 
Algorithm~\ref{alg:mpp} is the detailed pseudocode, and we present a succinct correctness proof.

\begin{algorithm}
\caption{SolveMeanPayoffParity}\label{alg:mpp}
\KwIn {A mean-payoff parity game $\Gamma = ((V,E,\langle V_1, V_2 \rangle,w, p)$ and objectives $\phi =
		\MeanPayoff(\V) \cap \Parity(p)$}
\KwOut {The winning region of player~1 in $\Gamma$, i.e., $W_1(\phi)$.}
\Begin{
\lIf {$V = \emptyset$}{\Return $\emptyset$}
$k \leftarrow \min\{p(v)  \mid  v \in V \}$, $i \leftarrow 0$, $A_0 \leftarrow V$\;
	\uIf{$k$ is even}{
        \lIf{$p$ has two priorities}{solve the mean-payoff Büchi game $\Gamma$
        with objectives $\phi$.}
	\Repeat{$A_i = A_{i+1}$}{
                $Y_i \leftarrow \attr_1(A_i \cap p^{-1}(k))$\;
                $C_i \leftarrow (A_i \setminus Y_i)$\;
		$X_i \leftarrow C_i\setminus \solvempp(\Gamma \restr C_i , \phi))$\;
		\uIf{$X_i \neq \emptyset$}{
                    $A_{i+1} \leftarrow A_i\setminus \attr_2(X_i)$\;
                    }\Else{
                        $U_i \leftarrow W_1(\varphi)$ in $\Gamma \restr A_i$,
                            where $\varphi = \MeanPayoff(\V)$\\
                        $X_i \leftarrow A_i \setminus U_i$.\\
                        $A_{i+1} = A_i \setminus \attr_2(X_i)$

                }
        $i \leftarrow  i+1$\;
	}
	\Return $A_i$\;
	} \ElseIf{$k$ is odd} {
        \lIf{$p$ has two priorities}{solve the mean-payoff Büchi game $\Gamma$
            with objectives $\phi$.}
            $B_0 \leftarrow V$\;
            \Repeat{$B_i = B_{i+1}$}{
                $Y_i \leftarrow \attr_2(B_i \cap p^{-1}(k))$\;
                $X_i \leftarrow B_i\setminus Y_i$\;
                $Z_i \leftarrow \solvempp(\Gamma \restr X_i, \phi) $\;
                $B_{i+1} \leftarrow B_i \setminus \attr_1(Z_1)$\;
                $i \leftarrow i+1$\;
            }
            \Return $V \setminus B_i$\;
	}

}
\end{algorithm}

\begin{lemma}[Correctness]
Given game graph $\Gamma = (V, E, \langle V_1, V_2 \rangle, w, p)$ 
with objectives $\phi = \MeanPayoff(\V) \cap \Parity(p)$,
Algorithm~\ref{alg:mpp} correctly computes the set $W_1(\phi)$ for $\Gamma$.
\end{lemma}

\begin{proof}
We proceed by induction on the number of priorities $d$.
For the base cases, i.e. when $d = 2$ we can use Theorem~\ref{thm:buchi}
and Theorem~\ref{thm:cobuchi} to receive the winning set $W_1(\phi)$ for $\Gamma$.
Assume now that for $d \leq k$ Algorithm~\ref{alg:mpp} correctly returns
$W_1(\phi)$ for $\Gamma$.
For the induction step, assume that we have $d=k+1$ priorities. We need to show
that we will correctly return $W_1(\phi)$ for $\Gamma$. Therefore we make a case
distinction whether the smallest priority is even or odd.
\begin{itemize}
	\item Assume the smallest priority $k$ in $\Gamma$ is even and thus $v$
        is in the set returned by line 18.
	By the construction of the algorithm, we have two cases:
	\begin{itemize}
		\item Either $v$ is in $Y_i$ (line~7),
		\item or $v$ is in the mean-payoff parity  
                winning set of the game $\Gamma \restr C_i$ for some $i$ (line~9).
	\end{itemize}
        Note that if $v$ is in $Y_i$ we can ensure $\MeanPayoff(\V)$ 
        because we remove every vertex not sufficing the objective in
        lines~13-15. We thus have a strategy $\sigma_{mp}$ ensuring
        $\MeanPayoff(\V)$.
	It remains to argue why we can in both cases win the mean-payoff parity game. 
        If $v$ is in the mean-payoff parity winning set obtained by the
        recursive call (line~9) we have a strategy by the induction hypothesis.
        If $v$ is in $Y_i$ we will propose a strategy for player-1 starting from $v$.
        The strategy will be played in rounds $1,2,\dots$. In round $i$ we will
        play the following strategy:
        Because $v$ is in the player-1 attractor of 
        $p^{-1}(k)$ we can visit $k$.
	This could mean that we accumulate (in the worst case) up to $-(n-1)W$ credits for the mean-payoff
	objective. If we end up in the mean-payoff parity winning set, we win by the induction hypothesis.
	If we are still in the player-1 attractor set $Y_i$, 
        we play $\sigma_{mp}$ for $i$ steps which will ensure the mean-payoff conditions
        as $i \rightarrow \infty$. 
	After playing $\sigma_{mp}$ for $i$ steps we can end up (i) again in 
        the player-1 attractor set $Y_i$, enabling us to visit a vertex of priority $k$, or  
        (ii) in the mean-payoff parity winning set where we win by
        induction hypothesis. 
	\item Assume the smallest priority $k$ in $\Gamma$ is odd. Thus $v$ must have been in the set 
	returned by line 29. It must be that $v$ is in some $Z_i$ or some
        player-1 attractor to it. If $v$ is in $Z_i$ we win the game by the
        induction hypothesis. Otherwise if we are in a
	player-1 attractor to $Z_i$ we will use the strategy induced by the attractor to reach $Z_i$.
\end{itemize}
\end{proof}

\begin{lemma}[Running Time]
The worst case complexity of Algorithm~\ref{alg:mpp} is $\O(n^{d-1} \cdot m \cdot W)$.
\end{lemma}
\begin{proof}
Let $T(n,d,m,w)$ be the complexity of Algorithm~\ref{alg:mpp}. Since every
recursive call removes at least one state from $A_i$ and since the number of
priorities decrease in a recursive call we get the following recurrence
relation: $T(n,d,m,w) = n(T(n,d-1,m,w) + \O(m)) + \O(nmW) = 
n(T,n,d-1,m,w) + \O(nm) + \O(nmW)$. Note that $\O(m)$ is used to calculate the
attractors. We only get $\O(nmW)$ once every iteration, because we can use the
decremental algorithm introduced in Section~\ref{sec:decr} to calculate the
mean-payoff objectives in line~13. Note that this is particularly possible because only
player-2 attractors get removed, thus ensuring the input condition for the
algorithm. We can thus simplify to $T(n,d,m,w) = 
n(T,n,d-1,m,w) + \O(nmW)$. Also note that when we have the base case $T(n,2,m,w)$
we can solve the problem using Theorem~\ref{thm:buchi} and
Theorem~\ref{thm:cobuchi} in $\O(n\cdot m \cdot w)$. Solving the recurrence
relation yields $T(n,d,m,w) = n^{d-1} + (d-1) \cdot  \O(n\cdot m \cdot W)$ which
concludes the proof.
\end{proof}

The last two lemmas yield Theorem~\ref{thm:tmpp}.


\end{document}